\documentclass[runningheads]{llncs}
\usepackage[margin=1in]{geometry}
\usepackage{cite}
\usepackage{amsmath, amsfonts, amssymb}
\usepackage{algorithmic}
\usepackage{graphicx}
\usepackage{textcomp}
\usepackage[linesnumbered, ruled]{algorithm2e}
\usepackage{enumerate}

\begin{document}

\spnewtheorem{assumption}{Assumption}{\bfseries}{\itshape}


\title{Closed Non-atomic Resource Allocation Games}

\author{Costas Courcoubetis\inst{1}
%
\and
Antonis Dimakis\inst{2}
\authorrunning{C. Courcoubetis, A. Dimakis}
\institute{Singapore University of Technology and Design
   \email{costas@sutd.edu.sg}\\
 \and
Athens University of Economics and Business
   \email{dimakis@aueb.gr}
}}

\maketitle

\begin{abstract}
How is efficiency affected when demand excesses over supply are signalled through waiting in queues?
We consider a class of congestion games with a nonatomic set of players of a constant mass, based on
a formulation of generic linear programs as sequential resource allocation games. 
Players continuously select activities such that they maximize linear objectives interpreted as time-average of activity rewards, while active resource constraints cause queueing. In turn, the resulting waiting delays enter in the optimization problem of 
each player.

The existence of Wardrop-type equilibria and their properties are investivated
by means of a potential function related to proportional fairness.
The inefficiency of the equilibria relative to optimal resource allocation is characterized  through the price of anarchy which is 2 if all players are of the same type ($\infty$ if not).
\end{abstract}




\section{Introduction}

In crowdsourcing, access, and sharing economies, a large number of individuals interact to exchange goods and services, with each individual pursuing his or her own interest. The matching of supply and demand takes place in shorter times than in traditional product-service economies, so 
mismatches may be manifested also in non-monetary terms as congestion. For example in ride-hailing, it is common for drivers to face significant waiting delays until they are matched with a customer, if the number of available drivers in an area exceeds local demand.  In this paper we consider a class of nonatomic games, and the appropriate equilibrium concept, which capture the noncooperative behavior and congestion effects in resource allocation settings such as above.

There is a large literature on congestion games~\cite{rosenthal}, which examine the interaction between congestion and noncooperative behavior. In the case where the number of players is large and each has a negligible effect on congestion, 
nonatomic congestion games view players as a continuous mass whose equilibrium behavior is described by
Wardrop-type equilibria, first studied for road traffic in~\cite{wardrop, beckmann}.
In this paper we consider a similar case but with a {\em constant} player mass playing a sequential game.

\begin{figure}[tp]
   \centering
   \includegraphics[scale=0.2]{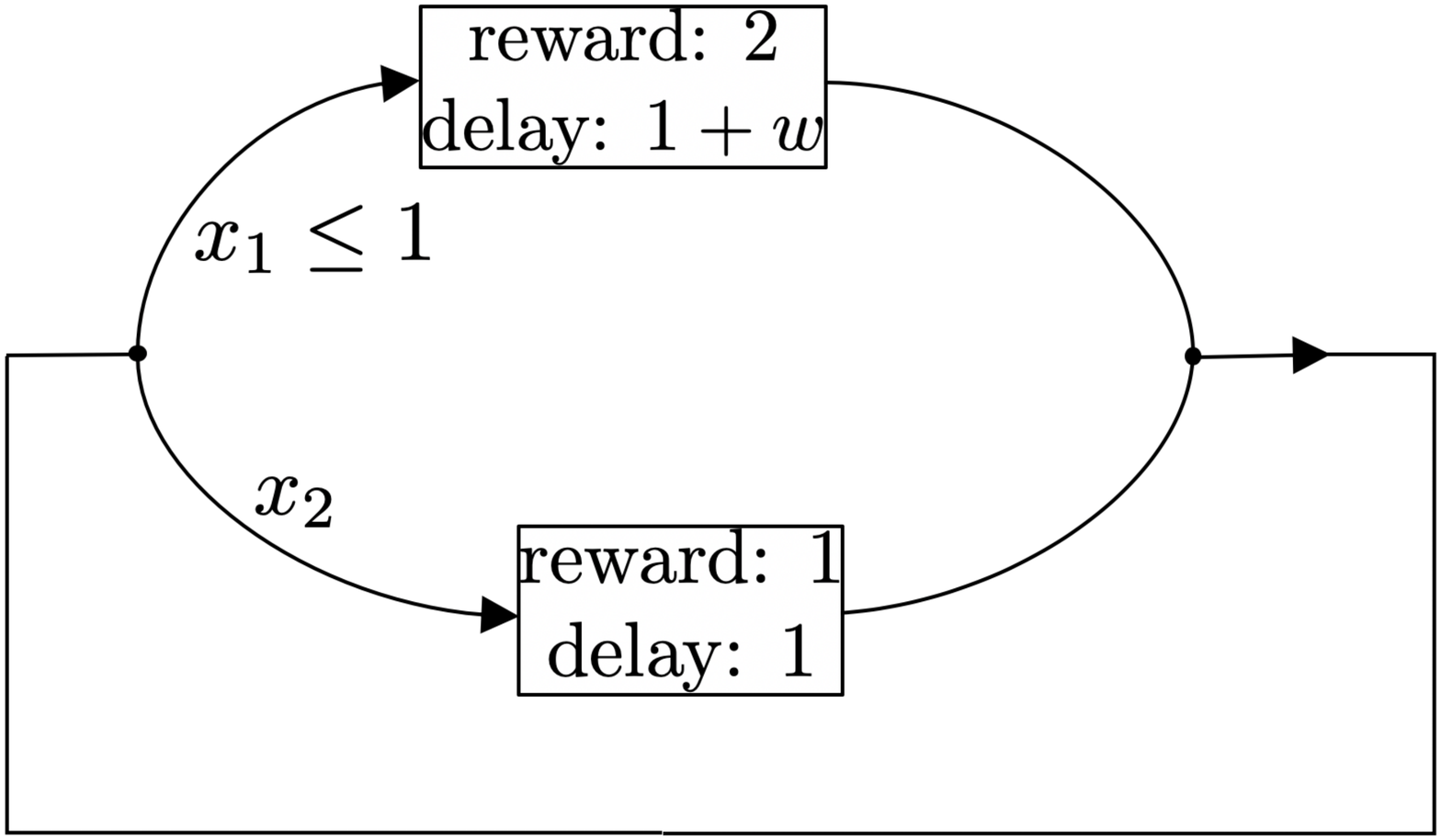} 
   \caption{Selfish circulation of a constant mass of nonatomic players which maximize their average reward per unit time. The waiting delay $w$ is 0 if $x_1<1$. }
   \label{fig:pigou}
\end{figure}

To illustrate the difference consider the example in Fig.~\ref{fig:pigou},
which resembles Pigou's selfish routing example~\cite{pigou}.
Players flow from left to right
utilizing routes 1 (upper) and 2 (lower). Traversal of the upper route offers reward 2 while the lower route a unit reward. 
Let $x_i$ be the rate of players flowing through route $i$, and assume the 
maximum rate at which players can flow through route is 1, i.e., $x_1\leq 1$. Route 2 is not rate limited and the delay is always 1.
The players do not exit the system after a traversal; they return back
to the origin on the left and keep circulating. 
Assume the mass of players in the system is equal to 2.
A basic difference with Pigou's example is that the delay of route 1 is not a function of the flow there. 
It is expressed as $1+ w$ where $w$ is the additional waiting time players face
if they accumulate in route 1 due to the limitation in the flow rate. 
If $x_1<1$, since the players behave as a fluid, they do not accumulate so $w = 0$. On the other hand, if
$x_1 = 1$ then $w$ is not a function of $x_1$ anymore; 
it also depends on the player mass currently on route 2, i.e., $x_2$.
The total mass is 2, so in route 1 there must be $2 - x_2$ players whose mass is also expressed by Little's law~\cite{little}
as $(1+w)x_1$.  Thus, $w = 1 - x_2$ if $x_1 = 1$.

What flows $x_1, x_2$ will result from a `selfish circulation'?  Players
will prefer the upper route as long as the reward per round-trip is higher than that offered by the lower route, i.e.,
$\frac{2}{1+w} \geq 1$.  
If $x_1=1$, this is always the case if $x_2>0$ since $w=1-x_2$ and hence $\frac{2}{1+w} = 1+\frac{x_2}{1+w}>1$.
In fact, the entire mass of players can utilize the upper route, in which case $w=1$, 
and still make that route at least as preferrable as the lower one. Thus, $x_1 = 1, x_2 = 0$ is an
equilibrium, and is easy to see that it is the only one. The  long-run average total reward resulting from the
equilibrium circulation of players is 2 (or 1 per unit of player mass).
This is less that the maximum possible reward, obtained by solving:
\begin{align}
\label{eq:lpexample}
\max \quad & 2x_1 + x_2\\
\nonumber
\text{such that } & x_1 \leq 1,\\
\nonumber
& x_1 + x_2 = 2,\\
\nonumber
\text{over } & x_1, x_2\geq 0,
\end{align}
(Note that in the mass constraint $x_1 + x_2 = 2$ we  do not need to account for waiting players as those can always
be assigned more profitably on route 2.)
Its (obvious) optimal solution, $x^*_1 = 1, x^*_2 = 1$, yields 3 as average reward (or 1.5 per unit mass),
as the use of route 2 increases the average reward by 50\%.

More generally, we take arbitrary linear programs with a single `mass constraint', similar to~\eqref{eq:lpexample}, as our point of departure, and consider a selfish circulation of nonatomic players which select activities that maximize their time-average rewards, given by the objective function. 
The inequality constraints correspond to resource constraints, which when active, cause waiting delays. 
The latter include all the relevant information a player needs to know about the other players' strategies in order to
maximize his or her own time-average rewards, by solving a dynamic program.

In this respect, our work is related to the stationary anonymous sequential games in~\cite{wiecek} where the 
players are aware of how the other players are distributed over strategies, and maximize their  time-average payoffs. 
The equilibria in these systems are similar to our Definition~\ref{def:eq} except
that we allow the inclusion of `balance' constraints, which are private to each player without 
additional private state variables
(see Section~\ref{sec:rag}).
Another difference with the literature on anonymous sequential games is that the existence of equilibrium there is established using nonconstructive compactness arguments (see~\cite{jovanovic, wiecek}).

In Section~\ref{sec:results} the existence and uniqueness properties are established by means of a potential function the players unknowingly maximize, which is markedly different from the objective of the linear program.
For example, the optimization problem corresponding to~\eqref{eq:lpexample} is:
\begin{align*}
\max\quad & 2\log(2x_1 + x_2) - x_1 - x_2\\
\text{such that } & x_1 \leq 1,\\
\text{over } & x_1, x_2\geq 0,
\end{align*}
with the sole optimal solution being the equilibrium flow $x_1 = 1, x_2 = 0$.
As the optimal solutions in the two optimization problems in general do not coincide, the time-average rewards in equilibria
will be strictly lower. The largest possible ratio of the maximum reward over the reward at equilibrium, called the {\em price of anarchy}, is a measure of the inefficiency of equilibrium, first proposed by Koutsoupias and Papadimitriou in~\cite{koutsoupias}.
For nonatomic congestion games, the price of anarchy has been first computed in~\cite{roughgarden} for various families of delay functions.
In Proposition~\ref{prop:poa} we establish that the price of anarchy is 2, attained in the limit 
of a sequence of simple examples, similar to that above.


In Section~\ref{sec:examples} we consider examples from three areas: ride-hailing, crowdsourcing platforms, and
interacting semi-Markov decision processes. 
For each case, we give example formulations as resource allocation games and obtain some new results.
In~\cite{dai} the authors consider the optimization of fluid model of a ride-hailing system, where
mismatches of demand and supply cause waiting delays similar to fluid queues, but no gaming aspects are explored.
This is done in~\cite{johari} where a concept of equilibrium similar to ours is defined for the two strategies of whether to circulate (with routing fixed) or not. \cite{bimpikis} considers routing as part of the strategy set and establishes 
existence of equilibrium in symmetric systems with identical players. 
Corollary~\ref{cor:delay} in Section~\ref{sec:results} extends the results
of~\cite{bimpikis} for arbitrary network topologies and multiple player types.

The statements and proofs of the main results
are given in Section~\ref{sec:results}, followed by discussion in 
Section~\ref{sec:conclusions}.


\section{A Linear Resource Allocation Game}
\label{sec:rag}
A mass $d_l>0$ of type $l=1,\ldots,L$ players generate value by performing a set of $J$ activities which consume $I$ resources.
For a type $l$ player, activity $j\in\{1,\ldots,J\}$ takes time $t^l_j\geq0$ to complete and consumes $a^l_{ij}\geq 0$ 
units of resource $i\in\{1,\ldots,I\}$.
Let $t_l = (t^l_j,j = 1,\ldots,J)$ be the column vector of activity durations, and $b = (b_i, i=1,\ldots,I)$ where
$b_i>0$ is the rate at which resource $i$ is provided.
Also, let $x^l_{j}$ denote the total rate at which each activity $j$ is taken by all players of type $l$,
and $x^l = (x^l_j, j = 1,\ldots,J)$ be the column vector of type $l$ rates. 
Then the resource constraints are expressed as $\sum_lA_lx_l \leq b$,
where $A_l$ is a $I\times J$ matrix with $a^l_{ij}$ in its $i$-th row and $j$-th column.
Each type $l$ player receives reward $c^l_j$ for completing activitity $j$. The total reward rate for all type $l$ players
is expressed as ${c_l}^T x_l$, where $c_l= (c^l_{j}, j = 1,\ldots,n)$ is a column vector.

The activities can be interdependent in the sense that the rates $x_l$ for type $l$ satisfy $K_l$ homogenous {\em balance constraints}, i.e., $H_lx_l = 0$, for some $K_l\times J$ matrix $H_l$ with the element of the $k$-th row and $j$-th column denoted by $h^l_{kj}$. Note that the resource constraints restrict the aggregate rates, whereas the balance constraints restrict the strategies of each player.

Next, we define the activity rates which correspond to optimal resource allocation.
\begin{definition}
A vector of activity rates $x^* = (x_l^*, l = 1,\ldots, L)$ is {\em optimal} if it maximizes the total reward rate, i.e.,
\begin{align}
\label{eq:primal}
\max \quad & \sum_l{c_l}^T x_l\\
\label{eq:resource_constraint}
\text{such that } & \sum_lA_lx_l \leq b,\\
\label{eq:primal_balance}
& H_lx_l = 0,\\
\label{eq:primal_mass}
&{ t_l}^T x_l = d_l,\quad l=1,\ldots,L,\\
\nonumber
\text{over } & x_l \geq 0, \quad l=1,\ldots,L.
\end{align}
\end{definition}


If instead players act selfishly, each maximizes his or her own average reward rate. As activities are assigned with no central coordination, players  may have to wait before they can commence high reward activities due to competition for the limited resources.
Thus, players need also take into account the waiting delay $w^l_j$ before each activity $j$ can commence.

\begin{definition}
\label{def:eq}
A pair $(x^o, w)$  where $x^o = (x^o_l, l = 1,\ldots, L), w = (w_l, l = 1,\ldots, L)$ is an {\em equilibrium} if:
\begin{enumerate}
\item 
$x^o_l=(x^l_j,j=1,\ldots,J)$ is an optimal solution of
\begin{align}
\label{eq:decomposition1}
\max \quad&  {c_l}^T x_l \\   
\label{eq:decomposition1_balance}
\text{such that } & H_lx = 0,\\
\label{eq:decomposition1_mass}
& {t_l}^Tx + {w_l}^Tx= d_l,\\
\nonumber
\text{over } & x\in\mathbb{R}_+^J,
\end{align}
for each $l=1,\ldots,L$.

\item 
\label{def:cnd2}
\begin{equation}
\label{eq:resource_constraints}
\sum_lA_lx^o_l \leq b.
\end{equation}

\item 
\label{def:cnd3}
$w_l= {A_l}^T \delta$ for a nonnegative column vector $\delta = (\delta_i, i = 1,\ldots, I)$ with $\delta_i = 0$ if $\sum_{j,l}a^l_{ij}x^{l}_{j} < b_i$.
\end{enumerate}
\end{definition}

In problem~\eqref{eq:decomposition1} the time-average reward is maximized from the perspective of a type $l$ player: time is split into either performing some activity or waiting for it (as suggested by \eqref{eq:decomposition1_mass}) while respecting the balance constraints. 
Each player solves an instance of~\eqref{eq:decomposition1} for an infinitesimal mass in the righthand-side of~\eqref{eq:decomposition1_mass}, but as
optimal solutions are homogeneous of degree 1 with respect to the mass constant, $x^o_l$ gives the equilibrium rates for the entire mass of type $l$ 
players.

Resource constraints~\eqref{eq:resource_constraints} are not part of the optimization in~\eqref{eq:decomposition1} as resource capacities $b$ and aggregate rates $x^o$ are not directly known to  the players; resource exhaustion is signaled through the waiting delays $w$ instead. Condition 3 in Definition~\ref{def:eq} requires the delays to be of a specific form which can be thought to result from the following posited mechanism: 
tickets granting usage of single resource $i$ units are handed out from a booth for the respective resource at rate $b_i$.
Players of type $l$ can start performing activity $j$ 
once they have collected all tickets for the resources required, i.e., $a^l_{ij}$ tickets for each resource $i$.
If $\delta_i$ is the delay to obtain a single resource $i$ ticket then the waiting delay $w^l_j$ for collecting all the required tickets for activity $j$ is $\sum_{i}a^l_{ij}\delta_i$.

We make the assumption that there always exist activity assignments with positive rewards for all types.
\begin{assumption}[Feasibility]
\label{as:cx}
There exist
nonnegative vectors $x_l, l=1,\ldots,L$ with $\sum_lA_lx_l \leq b, H_lx_l = 0, {c_l}^Tx_l>0$ for all $l$. 
\end{assumption}

Also, we assume that all types have the incentive to participate under all possible waiting delays. For example, this is the case when there is an outside option with a positive reward.
\begin{assumption}[Participation]
\label{as:participate}
For every $w_l\geq 0, l = 1,\ldots,L$, the maximum value ${c_l}^T x_l$ in~\eqref{eq:decomposition1} is strictly positive.
\end{assumption}

\section{Examples}
\label{sec:examples}

\subsection{Ride-hailing}
\label{sec:ridesharing}
In this section we formulate a model for ride-hailing which fits into the resource allocation framework.
In ride-hailing systems
a population of drivers transport customers to their destinations. 
We consider a geographical area which we assume it is divided into a finite set of regions. Let $b_{i}$ be rate at which customers arrive in region $i$, with a proportion $q_{ij}$ of them
requesting transport to region $j$. In the context of resource allocation, customers are seen as resources.

The drivers constitute the players which we assume are of a single type with mass $d$. There are two types of activities:  
i) a `busy' activity, where the driver transports a customer who has been picked up from $i$ to his destination, and ii) a `free' activity, where the driver chooses to move from $i$ to $j$ without carrying a customer. 
Here it is assumed drivers cannot pickup customers from different regions, and any customers which exceed the 
driver capacity in a region are lost.
Notice that the `busy' activity may involve waiting if the supply of drivers exceeds the rate of arriving customers. 
Also, it does not include the customer's destination as this is typically not known to the driver before agreeing to serve the customer.
A driver is compensated with $c_i>0$ per unit time for giving a ride originating from region $i$.
Thus, the busy activity brings an expected reward $c_i\sum_j q_{ij}t_{ij}$, with $t_{ij}>0$ being the transport time from $i$ to $j$, while the free activity is not rewarded and takes $t_{ij}$ time to complete.

Let $x_i$ be the rate of drivers choosing the busy activity in $i$, and $y_{ij}$ be the rate of drivers moving free from $i$ to $j$. Then, as the inflow and outflow of drivers in any region must balance, we have the constraint:
\begin{equation*}
x_i + \sum_j y_{ij} = \sum_{j} x_j q_{ji} + \sum_j y_{ji},
\end{equation*}
for each $i$. The first term on the righthand side consists of the rate of busy drivers arriving to $i$ after reaching the destination of the customer that was picked up from $j$, for any $j$.
Each busy activity `consumes' a customer, so we have the resource constraint $x_i \leq b_i$ for all $i$.

Corollary~\ref{cor:delay} guarantees the existence of an equilibrium $(x^o, w)$, where the average reward and
the mass of waiting drivers have unique values in all equilibria. The equilibria can be computed by solving the convex optimization problem~\eqref{eq:functional}.

\begin{example}
\label{eg:ridesharing}

Consider an area with three regions as depicted in Fig.~\ref{fig:regions}.
\begin{figure}[tp]
   \centering
   \includegraphics[scale=0.2]{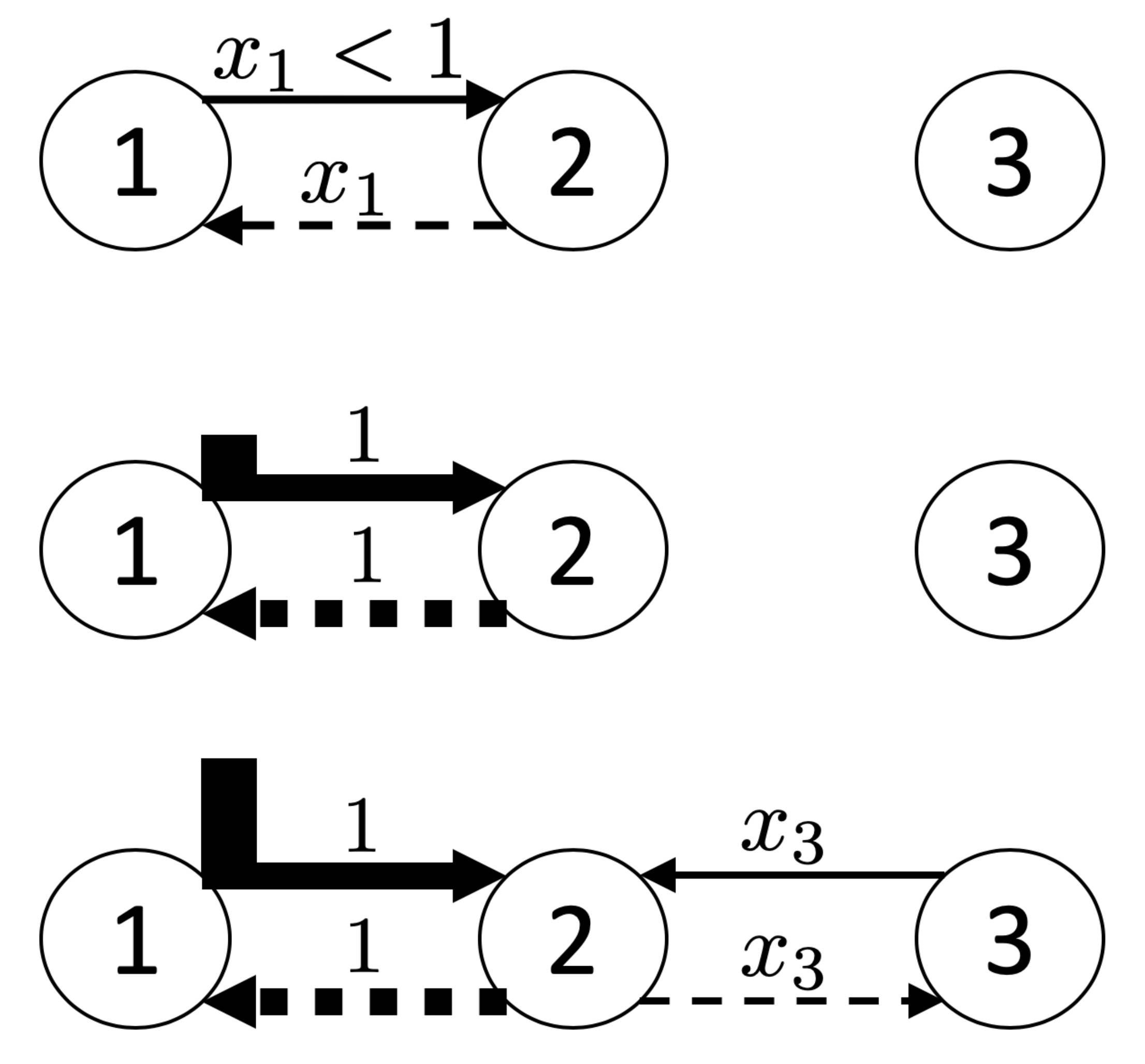} 
   \caption{A ride-hailing system with $d$ drivers comprised by three regions (depicted as arcs). Top: busy and free flows (solid and dashed arrows, respectively) for $d\leq2$. Middle: for $2<d\leq 4$ drivers find it still more profitable to queue in region 1 instead of serving region 3. Bottom: for $4 < d \leq 6$ the queueing delay is sufficiently high so drivers start serving region 3 which became equally profitable due to the high waiting in region 1.}
   \label{fig:regions}
\end{figure}
Customers request transport from regions 1, 3, with unit rate from each, towards the center region 2, i.e., 
$b_1 = b_3 = 1, b_2 = 0, q_{12} = q_{32} = 1$. For each trip transporting a customer from region 3, the driver receives a unit reward, while from region 1 receives double ($c_1 = 2, c_3 = 1$).
Assume unit transport times between neighboring regions: $t_{ij} = |i - j|$ for any two regions $i,j$. 

We observe three regimes, depending on the mass $d$ of drivers.
In the first regime, there is no waiting in our fluid model to pick up customers in region 1.
Serving continuously customers from 1 to 2 generates 2 units of reward per round-trip, i.e., an average reward rate of 1.
This is higher than the $1/2$ average reward rate earned from serving customers from 3 to 2. Thus, all $d$
drivers will choose serving region 1 provided they can always pickup a customer on their return to region 1. This will be possible as long
there are no waiting drivers in region 1, and so $d$ equals the number of drivers $2x_1$ 
on the forward and return trip. As $x_1\leq b_1 = 1$, we must have $d\leq 2$.

In the second regime, queues start foming in region 1, but are not long enough to motivate drivers to serve customers in region 3. If $d$ is just above $2$, then the customer demand from 1 cannot support all drivers and so some of them may wait. They will do so
if the average reward (including wait) is less than the average reward serving region 3 (in which there is no waiting.)
At this point the total reward rate is $1$ and does not increase for small increases of $d$, even though the revenue stream from customers from region 3 is not utilized. Clearly, this equilibrium does not maximize the total reward 
rate~\eqref{eq:primal}, and even more, it is not Pareto efficient, i.e.,
the society of drivers as a whole could gain more by serving region 3 too and splitting the total proceeds.

In the third regime, serving region 3 becomes a best choice due to the high delays in waiting for customers in region 1.
The mass $d-2$ of waiting drivers grows, as $d$ is further increased, until the average reward $2/d$ 
equals the 1/2 reward for serving region 3, i.e., $d = 4$. If $d>4$, the extra $d-4$ drivers all serve region 3 while the
queue at 1 stops increasing until $d>6$ in which point queues in both regions will increase at the same time while keeping
the same average rewards.

\end{example}

In the above example, as the one in Fig.~\ref{fig:pigou}, 
the optimal rewards given by~\eqref{eq:primal} and those at equilibrium deviate. 
In Proposition~\ref{prop:poa} we show that the optimal value cannot exceed twice the reward at equilibrium. 


\subsection{Crowdsourcing}
\label{sec:crowdsourcing}
In crowdsourcing platforms, tasks which typically form small parts of a much larger effort, are executed by many participants in parallel which may receive a reward for each task completion. The tasks vary in their difficulty,  time to complete, reward given etc., and so do the capabilities and task preferences of the participants. The latter, typically select tasks in order to receive as high rewards as possible.

We can formulate a simplified model in terms of resource allocation as follows: 
tasks of type $i$  are generated at rate $b_i$ and correspond to a unit of resource $i$. Activity $i$ concerns the processing of one type $i$ task. Participants, which are the players here, are of $L$ different types, with the rewards $c_l$, task processing times $t_l$ being dependent on the type $l$. 
One of the activities corresponds to idling and has 0 reward.
Task types which cannot be undertaken by a participant type, are assumed to bring a negative reward so that
they are never selected.
$A_l$ are all unit matrices and there are no balance constraints as the tasks are assumed independent.

Theorem~\ref{thm:main} below, implies an equilibrium $(x^o, w)$ exists and the resulting aggregate reward ${c_l}^Tx^o_l$ attained by type $l$ participants, for each $l$, is unique. By Corollary~\ref{cor:delay},  the waiting delay $w^l_j$ is uniquely determined for tasks with nonzero equilibrium rates, and given by~\eqref{eq:wformula}.

How does the total reward $\sum_l {c_l}^Tx^o_l$ compares to the  maximum possible reward when task assignment is performed by~\eqref{eq:primal}, with
the same participants? Notice that if all rewards of one participant type, e.g., 1, are doubled while those of the other types do not change, $x^o$ remains an equilibrium as the relative rewards between activities matter in players' selection; not the actual rewards. Thus, there will be no increase in the amount of tasks completed by type 1. This is not the case under optimal task assignment, as the change will likely allow type 1 to complete more tasks (by having other types idle) because their rewards are part of the system objective. The increase in optimal rewards may be arbitrarily larger than the increase in ${c_l}^Tx^o_l$, as illustrated in the following example.
\begin{example}
Consider $L = 2$ participant types, with tasks of a single type (besides the idling task) arriving at rate 1. The rewarded value is $1/\epsilon$ and
$1$ for type 1 and 2 respectively, for some constant $\epsilon>0$.  The mass of type 2 participants is $1/\epsilon$, while that of type 1 is unit. Hence,
\begin{align*}
\max\quad & \frac{1}{\epsilon} x_1^1 + x_1^2\\
\text{s.t. } & x_1^1 + x_1^2 \leq 1,\\
& x_1^1 + x_2^1 = 1,\\
& x_1^2 + x^2_2 = \frac{1}{\epsilon},\\
\text{over }& x^1_1, x^1_2, x^2_1, x^2_2\geq 0.
\end{align*}
yields the optimal solution $x^1_1 = 1, x_1^2 = 0$, with $x^1_2, x^2_2$ being the rates of the idling activities for each participant type. 
This is expected, as the higher value type 1 participants generate more value than type 2, attaining total value $1/\epsilon$ and both types choose task type 1 since it is the only one generating positive revenue.

On the other hand, the (unique) equilibrium has
\begin{equation*}
x^1_1 = \frac{\epsilon}{1+\epsilon},\> x^2_1 = \frac{1}{1+\epsilon}\>, w_1 = \frac{1}{\epsilon}
\end{equation*}
yielding total value $2/(1+\epsilon)$. This can be formally shown either directly from Definition~\ref{def:eq}, or by Theorem~\ref{thm:main} below, but it is 
expected because type $1$ participants are a fraction $\epsilon$ of type 2.

The ratio of the optimal to the equilibrium value, i.e., the `price of anarchy', is
$\frac{1}{2} + \frac{1}{2\epsilon}$ which approaches $\infty$ as $\epsilon\rightarrow 0$.

\end{example}

\subsection{Interacting semi-Markov Decision Processes}
\label{sec:mdps}
In this section we formulate a nonatomic game with players' states evolving
according to semi-Markov decision processes (SMDPs), which interact through congestion effects due to linear
constraints. Although Proposition~\ref{prop:mdp} below holds for players with SMDPs of multiple types, we state it for a single player type to simplify notation.
We then show that such games are instances of stationary anonymous sequential games~\cite{jovanovic, wiecek}.

Consider an SMDP with a finite state space $\mathcal{S}$ and action space $\mathcal{A}$. At each state
$i\in\mathcal{S}$, action $a\in\mathcal{A}$ will make the process transit to $j$ with probability $p_{ij}^a$ after a random time with mean $t_{ia}>0$, which is independent of the past conditionally on the current state and action.
A stationary policy is specified by the probability $p(a|i)$ of choosing action $a$ once transitioning to $i$, for every $i\in\mathcal{S}, a\in\mathcal{A}$.
We assume the transition probabilities are such that the embedded Markov chain resulting from any stationary policy irreducible, so in particular the SMDP possesses a unique stationary distribution $(\pi_i, i\in \mathcal{S})$. Under this distribution, let $x_{ia}=d\pi_ip(a|i)$ be the average rate at which action $a$ is taken in state $i$, by $d$ copies of the SMDP, all following the same policy.

Let $c_{ia}$ be the reward received for taking action $a$ in state $i$.
Action rates $x = (x_{ia}, i\in\mathcal{S}, a\in\mathcal{A})$ are constrained by resource constraints of the form $Ax\leq b$, 
for nonnegative $I\times \left|\mathcal{S}\times\mathcal{A}\right|$ matrix $A$ and column vector $b$. 
As in the general framework, active resource constraints cause a waiting delay $w_{ia}$ before 
action $a$ in state $i$ can be taken.

We consider equilibria of the following form.
\begin{definition}
$(p, x^o, w)$ is an {\em SMDP equilibrium} if and only if:
\begin{enumerate}
\item
The policy $p = (p(a|i), i\in \mathcal{S}, a\in\mathcal{A})$ solves the dynamic programming equation
\begin{equation}
\label{eq:dp}
V(i) = \max_{a\in {\mathcal A}}\left[ c_{ia} + \gamma\left(t_{ia} +w_{ia}\right) + \sum_{j\in \mathcal{A}}p_{ij}^a V(j)\right],\> i\in \mathcal{S},
\end{equation}
i.e., the maximum is attained for any $a$ with $p(a|i)>0$.
\item
\begin{equation}
\label{eq:xia}
x^o_{ia} = d\pi_i p(a|i),\> i\in\mathcal{S}, a\in\mathcal{A},
\end{equation}
where $(\pi_i, i\in\mathcal{S})$ is the stationary distribution under policy $p$.

\item
$Ax^o\leq b$,

\item
$w = {A}^T \delta$ for a nonnegative column vector $\delta$ with $\delta_q = 0$ if $\sum_{i,a}a_{q,ia}x^o_{q,ia} < b_q$.

\end{enumerate}
\end{definition}

\begin{proposition}
\label{prop:mdp}
An SMDP equilibrium $(x^o, w, p)$ exists and the time-average reward attained by the policy $p$ is the same in every  equilibrium.
\end{proposition}
\begin{proof}
SMDP equilibria directly correspond to equilibria of Definition~\ref{def:eq}, as~\eqref{eq:dp} is equivalent to~\eqref{eq:decomposition1}
for $L=1$, and set of activities $\mathcal{S}\times\mathcal{A}$.
It is well known that the stationary policies $p$ which optimize~\eqref{eq:dp} correspond to optimal solutions $(y^*_{ia}, (i,a)\in \mathcal{S}\times\mathcal{A})$ of the linear program~\cite{kallenberg}:
\begin{align}
\nonumber
\max \quad & \sum_{(i,a)\in\mathcal{S}\times\mathcal{A}} c_{ia}y_{ia} \\
\label{eq:balance_mdp}
\text{s.t. } & \sum_{a\in \mathcal{A}}y_{ia} = \sum_{(j,a')\in \mathcal{S}\times\mathcal{A}} y_{ja'} p_{ji}^{a'},\> i \in \mathcal{S},\\
\label{eq:mdp_mass}
& \sum_{(i,a)\in \mathcal{S}\times\mathcal{A}}\left(t_{ia} + w_{ia}\right)y_{ia} = 1,\\
\nonumber
\text{over }& y_{ia}\geq 0, (i,a)\in \mathcal{S}\times\mathcal{A},
\end{align}
where $y^*_{ia}$ corresponds to the rate action $a$ is chosen at $i$ under
an optimal policy of~\eqref{eq:dp}. This implies $x^o/d$ is an optimal solution, since it corresponds to $p$, by~\eqref{eq:xia}.
Therefore, $(x^o, w)$ satisfies the conditions in Definition~\ref{def:eq} for balance constraints given by~\eqref{eq:balance_mdp}.

The converse is also true, as given an equilibrium $(x^o, w)$,
\begin{equation*}
p(a|i) = \frac{x^o_{ia}}{\sum_{a'}x^o_{ia'}},\> 
\pi_i = \frac{\sum_a(t_{ia} + w_{ia})x^o_{ia}}{\sum_{j,a'} (t_{ja'} + w_{ja'})x^o_{ja'}}
\end{equation*}
define a policy and the corresponding stationary distribution which give an SMDP equilibrium. From Corollary~\ref{cor:delay}, an SMDP equilibrium exists.\qed

\end{proof}

If no resource constraint is active then $w_{ia}=0$ for all $i, a$, and no interaction takes place between the SMDPs.  
In this case the equilibrium policies achieve the maximum possible total average reward, and a joint policy selection
(control centralization) cannot produce a higher total reward. 
If some constraints are active in equilibrium and waiting results then the average 
reward is strictly below the one possible under centralized control. 
This drop due to decentralization, cannot be more than half
because of Proposition~\ref{prop:poa} below. 

\subsubsection{Relation to stationary anonymous sequential games:}

In stationary anonymous sequential games~\cite{jovanovic, wiecek}, each player knows its own state and the distribution
$n = (n_{ia}, i\in\mathcal{S}, a \in\mathcal{A})$ of player mass on state-action pairs.
The game between SMDPs is an instance of a (nonlinear) stationary anonymous sequential game because 
the information on the aggregate, $(x,w)$, and $n$ are equivalent, through the identities
$n_{ia} = \left(t_{ia} + w_{ia}\right)x_{ia}$ for each $i,a$.

\begin{lemma}
For each nonnegative $n = (n_{ia}, i\in\mathcal{S}, a \in\mathcal{A})$ with 
$\sum_{(i,a)\in\mathcal{S}\times\mathcal{A}} n_{ia} = d$, there exist unique
$x(n) = (x_{ia}(n), (i,a)\in\mathcal{S}\times\mathcal{A})$, and
$w(n) = A^T\delta$ with $\delta\in\mathbb{R}_+^I$, such that
\begin{equation}
\label{eq:lemcnd}
\begin{array}{c}
n_{ia} = (t_{ia} + w_{ia}(n))x_{ia}(n),\text{ for all }i\in\mathcal{S}, a\in\mathcal{A},\\
Ax(n)\leq b,\quad \delta^T\left(Ax(n)-b\right) = 0.
\end{array}
\end{equation}
The mapping $n\mapsto (x(n), w(n))$ is continuous.
\end{lemma}
\begin{proof}
Consider the optimization problem
\begin{align*}
\max\quad & \sum_{i,a} \left[n_{ia} \log(x_{ia}) - t_{ia}x_{ia}\right]\\
\text{s.t. } & Ax\leq b,\\
\text{over }& x = (x_{ia}, i\in\mathcal{S}, a \in\mathcal{A})\geq 0.
\end{align*}
A unique solution $x(n)$ exists, as the objective is a strictly concave function maximized over a set 
with compact closure and $x_{ia}(n)>0$ unless $n_{ia} = 0$. By strong duality, \eqref{eq:lemcnd}
characterizes the optimal solution with $\delta$ being the optimal solution of the dual problem.

The mapping $n\mapsto x(n)$ is continuous because the objective is continuous in $n$ and $x(n)$ is unique. The
continuity of $w(n)$ follows from~\eqref{eq:lemcnd}. \qed
\end{proof}

Since the action delay $w_{ia}(n)$ of 
a player taking action $a$ in state $i$, are continuous in $n$, the existence of equilibrium in Proposition~\ref{prop:mdp}
also follows from the
time-average reward case in~\cite{wiecek}. In Theorem~\ref{thm:main} we give a constructive proof which also 
yields uniqueness, based on a potential function for the game.

In the case $p^a_{ii} = 1$ for all $i, a$, the SMDP game becomes a finite strategy nonatomic (one-shot) 
game~\cite{schmeidler, mascolell} with the payoff of playing strategy $(i,a)$ given by
\begin{equation*}
\frac{c_{ia}}{t_{ia} + w_{ia}(n)}
=\frac{c_{ia}{x_{ia}(n)}}{n_{ia}},\quad i\in\mathcal{S},a\in\mathcal{A}.
\end{equation*}
The second case in Corollary~\ref{cor:delay} states that the equilibrium $w_{ia}(n)$ is unique if $x_{ia}(n)>0$. 


\section{Main Results}
\label{sec:results}


\subsection{Equilibrium}


Equilibria have the following variational characterization.

\begin{theorem}
\label{thm:main}
$(x^o, w)$ is an equilibrium if $x^o$ maximizes 
\begin{align}
\label{eq:functional}
\max \quad& \sum_l\left[d_l\log \left({c_l}^T x_l\right) -  {t_l}^T x_l\right]\\
\label{eq:capacity_main_thm}
\text{s.t. } & \sum_lA_lx_l \leq b,\\
\label{eq:balance_main_thm}
& H_lx_l = 0,\> l = 1,\ldots, L,\\
\nonumber
\text{over } & x_l\in\mathbb{R}_+^J,\> l = 1,\ldots, L,
\end{align}
and $w =( {A_l}^T \lambda, l=1,\ldots, L)$, where $\lambda\in\mathbb{R}^I_+$ are optimal values for the dual variables of the constraint~\eqref{eq:capacity_main_thm}.

Under Assumption~\ref{as:participate}, for any equilibrium $(x^o,w)$, $x^o$ maximizes~\eqref{eq:functional} and $w$ is as above.
\end{theorem}
\begin{proof}
Let $(x^o,w)$ be an equilibrium. Since $x^o$ maximizes~\eqref{eq:decomposition1}, Assumption~\ref{as:participate}
implies it  is also the maximizer of 
$d_l\log\left({c_l}^Tx_l\right)$
under the same constraints. 
The optimality conditions for this problem are:
\begin{enumerate}
\item (Feasibility) 
\begin{equation}
\label{eq:pf1}
H_lx^o_l = 0,\quad {t_l}^Tx^o_l + {w_l}^T x^o_l = d_l,
\end{equation}
\item (First order conditions)
\begin{equation}
\label{eq:foc1}
\frac{d_l c^l_j}{{c_l}^T{x^o_l}}
-t^l_j\nu_l -w^l_j\nu_l - \sum_k \mu^l_k h^l_{kj}\leq 0,\text{ with equality if }x^l_j>0,
\end{equation}
\end{enumerate}
for some values $\nu_l, \mu^l_k$ of the dual variables, for each $l = 1,\ldots,L, k = 1,\ldots,K_l$.

First note that multiplying both sides of~\eqref{eq:foc1} with $x^l_j$, summing over $j$, and applying~\eqref{eq:pf1} yields $\nu_l = 1$. For $\delta$ as in Definition~\ref{def:eq}, letting
$\lambda = \delta$ yields
\begin{equation}
\label{eq:lpf}
\lambda \geq 0, \lambda^T\left(\sum_lA_lx^o_l - b\right)=0\,
\end{equation}
as well as $w_l = {A_l}^T\lambda$.
Thus, the conditions above are reexpressed as
\begin{enumerate}[1$'$)]
\item $H_lx^o_l = 0$,
\item
\begin{equation*}
\frac{d_l c^l_j}{{c_l}^T{x^o_l}}
-t^l_j - \sum_i  \lambda_i a^l_{ij}  -\sum_k \mu^l_k h^l_{kj}\leq 0,\text{ with equality if }x^l_j>0.
\end{equation*}
\end{enumerate}
These, along with~\eqref{eq:resource_constraints},\eqref{eq:lpf} are the optimality conditions for the problem~\eqref{eq:functional}, which $x^o$ satisfies.

By proceeding in the reverse direction, it is easy to see that 1$'$), 2$'$), \eqref{eq:resource_constraints},\eqref{eq:lpf} 
imply 1), 2), and so $(x^o,w)$ is equilibrium for $w$ as in the statement of the theorem. \qed

\end{proof}

The linear term $\sum_l {t_l}^Tx_l$ in~\eqref{eq:functional}, which we refer to as the {\em active mass},  corresponds to the total player mass that is engaged into any activity.

\begin{corollary}
\label{cor:delay}
\begin{enumerate}
\item 
Under Assumption~\ref{as:cx} there exists an equilibrium.

\item 
Under Assumption~\ref{as:participate},
the value ${c_l}^T{x^o_l}$ rewarded to type $l$ players and the active mass $\sum_l {t_l}^Tx^o_l$, assume the same values in all equilibria.

\item 
Under Assumption~\ref{as:participate},
if type $l$ has no balance constraints, i.e., $H_l=0$, the waiting delay $w^l_j$ is uniquely determined by
\begin{equation}
\label{eq:wformula}
\frac{c^l_j}{w^l_j + t^l_j} = \frac{{c_l}^T {x_l^o}}{d_l},
\end{equation}
whenever $x^l_j>0$ in any equilibrium.
\end{enumerate}

\end{corollary}

\begin{proof}
By Assumption~\ref{as:cx} the feasible set of problem~\eqref{eq:primal} is nonempty, and it has a compact closure. 
Since the objective function is continuous, a maximizing $x^o = (x^o_l,l=1,\ldots,L)$ exists inside the closure. As
the value of the objective function tends to $-\infty$ as ${c_l}^T{x_l}\rightarrow 0^+$, $x^o$ is feasible and so it is optimal.

Under Assumption~\ref{as:participate} any equilibrium $x^o$ corresponds to an optimum solution of~\eqref{eq:functional}. 
Since the objective function is strictly concave with respect to ${c_l}^T{x_l}$, the ${c_l}^T{x^o_l}$ values are unique.
As all equilibria yield the same optimal value in~\eqref{eq:functional}, $\sum_l {t_l}^Tx^o_l$ is also unique.

Equation~\eqref{eq:wformula} holds because $x^l_j>0$ results from~\eqref{eq:decomposition1} only if activity $j$'s reward per unit time, appearing on the lefthand side in~\eqref{eq:wformula}, is equal to the optimal one for type $l$ on the right. 
(This is a restatement of condition~\eqref{eq:foc1}.)
For every $l,j$ with $x^l_j>0$, $w^l_j$ is unique because ${c_l}^T{x^o_l}$ is. \qed
\end{proof}

If the maximization in~\eqref{eq:functional} is restricted to a constant active mass  (by including the constraint $\sum_l {t_l}^Tx_l=d$ for some
$d$) then only the first term, the aggregate of logarithmic rewards, is optimized.
This objective induces
a {\em proportionally fair} \cite{kelly} distribution of value between players, i.e., any changes to activity rates incur an aggregate of proportional value changes which is nonpositive.
Therefore, the value distribution at equilibrium is the proportionally fair allocation
under the additional restriction that the active mass is that at equilibrium, i.e., $\sum_l {t_l}^T x_l = \sum_l {t_l}^T x^o_l$.
 Note also that proportionally fair allocations coincide with
 the Nash bargaining solution if disagreement entails nonparticipation.

In the single player type case, equilibria achieve maximum value  when  only the active mass at equilibrium is allowed to participate.
Let $F(d')$ be the optimal value of~\eqref{eq:primal} for player mass $d'$, i.e.,
\begin{align}
\label{eq:Fproblem}
F(d') = \max\quad & c^T x\\
\nonumber
\text{s.t. } & Ax \leq b,\\
\nonumber
& Hx = 0,\\
\label{eq:Fmass}
& t^T x = d',\\
\nonumber
\text{over } & x \geq 0,
\end{align}
where we have dropped the type index.

\begin{corollary}
\label{cor:optimal_eq}
Let Assumption~\ref{as:participate} hold, and $F(d)>0$. 
For a single player type with mass $d$, the equilibrium $x^o$ is optimal for a player mass equal to the active mass at equilibrium, $t^Tx^o$, i.e., 
$c^T{x^o} = F(t^T x^o)$.


Moreover, the active mass at equilibrium is the unique $d'\in (0, d]$ with the property $F(d')=F'(d')d$, where $F'(d')$ is a subgradient at $d'$.
\end{corollary}
\begin{proof}
By Theorem~\ref{thm:main}, $x^o, d^o = t^T x^o$  maximize 
\begin{align}
\label{eq:smp}
\max \quad & c^T x e^{-\frac{d'}{d}}\\
\nonumber
\text{s.t. } & Ax \leq b,\\
\nonumber
& Hx = 0,\\
\nonumber
& t^Tx = d',\\
\nonumber
\text{over }& x \geq 0, d'\geq 0.
\end{align}
If $x$ is feasible in~\eqref{eq:Fproblem} for $d' = d$ then $xd'/d$ is also feasible for any $d'\leq d$. As $F(d)>0$, 
the feasible set of~\eqref{eq:Fproblem} is nonempty for any $d'\in (0,d]$, and $F(d')>0$ for all $d'>0$.
For $d'\in(0, d]$ fixed, optimizing~\eqref{eq:smp} with respect to $x$ yields the optimal value $F(d')e^{-\frac{d'}{d}}$ which itself is maximized for $d'=d^o$, and so
$c^T{x^o} = F(d^o)$ as well.

Now,
\begin{equation*}
-\frac{d'}{d} + \log F(d')
\end{equation*}
is a concave function of $d'$ with the maximizing $d'$ characterized by $F(d') = F'(d')d$ for a subgradient $F'(d')$. 
As $d^o$ is the unique maximizer, this equation identifies $d^o$ uniquely. \qed
\end{proof}

\subsection{Price of Anarchy}

For a single player type, we calculate the price of anarchy, i.e., the largest possible ratio of the optimal value and value at equilibrium,
\begin{equation*}
\sup_{\substack{
d>0, c\in\mathbb{R}^{J}, A\in\mathbb{R}_+^{I\times J}, b\in\mathbb{R}_+^I, H\in\mathbb{R}^{K\times J}, t\in\mathbb{R}_+^J, I, J, K\in\mathbb{N}\\
\text{s.t.\ Assumption~\ref{as:participate}, $F(d)>0$ hold}
}}
\frac{c^T{x^*}}{c^T {x^o}},
\end{equation*}
where $x^*$ is optimal, and $x^o$ an equilibrium, for player mass $d$.  Assumption~\ref{as:participate} and
$F(d)>0$ are used to ensure the numerator and denominator are positive so the ratio makes sense.
(For multiple player types the price of anarchy is infinite, as shown in the example of Section~\ref{sec:crowdsourcing}.)

\begin{proposition}
\label{prop:poa}

The  price of anarchy is 2.
\end{proposition}
\begin{proof}
Since $F$, defined in~\eqref{eq:Fproblem}, satisfies $F(d^o) = F'(d^o)d = c^T{x^o}$ by Corollary~\ref{cor:optimal_eq}, 
\begin{equation*}
c^T{x^*} -  c^T{x^o} = F(d) - F(d^o) \leq \frac{c^T{x^o}}{d}\left(d-d^o\right),
\end{equation*}
by using also the concavity of $F$.
Thus,
\begin{equation*}
\frac{c^T {x^*}}{c^T {x^o}} \leq  2 - \frac{d^o}{d} \leq 2,
\end{equation*}
as $d^o\leq d$.

To get a lower bound, for any $\epsilon>0$ consider the following instance of~\eqref{eq:primal}.
\begin{align*}
\max\quad & \frac{1}{\epsilon}x_1 + x_2\\
\text{s.t. } & x_1 \leq \epsilon,\\
& x_1 + x_2  = 1,\\
\text{over }& x_1, x_2\geq 0.
\end{align*}
The maximum value is $2-\epsilon$ achieved at $x^*_1 = \epsilon, x^*_2 = 1-\epsilon$.

On the other hand, waiting delays $w_1 = \frac{1}{\epsilon} - 1, w_2 = 0$,
induce $x^o_1 = \epsilon, x^o_2 = 0$ as optimal solution of
\begin{align*}
\max\quad & \frac{1}{\epsilon}x_1 + x_2\\
\text{s.t. } 
& x_1 + x_2  + w_1x_1 + w_2 x_2= 1\\
\text{over }& x_1, x_2\geq 0,
\end{align*}
which saturates the resource contraint $x^o_1 =\epsilon$. Thus, $x^o_1 = \epsilon, x_2^o = 0$ is the equilibrium
with value 1 and the $c^Tx^*/c^Tx^o$ ratio in this case is $2-\epsilon$ where $\epsilon>0$ arbitrarily small. \qed
\end{proof}

A way to force players pick activities which maximize total value as opposed to individual rewards, 
is to use the shadow prices $\lambda^*$ of the resource constraints in~\eqref{eq:primal} as resource prices. Under these `optimal' prices, activity $j$ has net reward $c_j - \sum_i a_{ij}\lambda_i^*$.

Now, duality implies the optimal $x^*$ maximizes the Lagrangian of~\eqref{eq:primal},
\begin{align*}
\max \quad& \left(c^T -{\lambda^*}^T A\right)x\\
\text{such that } & Hx = 0,\\
& t^Tx = d,\\
\text{over } & x \geq 0,
\end{align*}
and so $(x^*,0)$ is an equilibrium under optimal pricing, as $Ax^*\leq b$ also holds.

We include this in the following result.
\begin{proposition}
\label{prop:optimal_pricing}
Under optimal pricing, i.e., imposing a price $\lambda^*_i$ per unit of each resource $i$ where $(\lambda^*_1, \ldots,\lambda^*_I)$ are the optimal dual variables for the resource constraint in~\eqref{eq:primal}, the ensuing 
equilibrium yields the same value as the optimal value in~\eqref{eq:primal}.

The value at equilibrium without optimal pricing is at least as high as the net value retained by players under optimal pricing, i.e.,
\begin{equation*}
c^T {x^o}\geq c^T{x^*} -{\lambda^*}^T Ax^*.
\end{equation*}
\end{proposition}
\begin{proof}
Corollary~\ref{cor:optimal_eq} and the concavity of $F$ yield,
\begin{equation*}
c^T {x^o} = F'(d^o)d \geq F'(d)d = c^T{x^*} - {\lambda^*}^Tb
= c^T{x^*} -{\lambda^*}^T Ax^*,
\end{equation*}
where the second equality is by strong duality for problem~\eqref{eq:Fproblem}. \qed


\end{proof}


\section{Discussion}
\label{sec:conclusions}
In ordinary, i.e., one-shot, congestion games waiting delays have the role of congestion cost, usually given 
exogenously~\cite{rosenthal} or caused by randomness in the arrivals and service times~\cite{network_routing}. 
The delays in equilibrium  correspond  exactly to the Lagrange multipliers of flow balance constraints of an optimization problem maximizing the potential function of the game, e.g., see~\cite{network_routing}. This is also what happens in
Theorem~\ref{thm:main} where the delays are the Lagrange multipliers of the resource constraints in~\eqref{eq:functional}. 
Of course, this is a subsequence of 
how delay is defined in the third condition of Definition~\ref{def:eq} which is the complementary slackness condition for these constraints.
What is novel, to the best of the authors' knowledge, is the use of the potential function in~\eqref{eq:functional}
for sequential congestion games, where the delays are determined endogenously by constraints on player mass (i.e., Little's law~\cite{little}). 
In particular, the concavity of the potential function can be used in showing that the best response dynamics coupled with the waiting delay dynamics due to queueing,
\begin{equation*}
\dot{\delta}_i = \frac{1}{b_i}\sum_{l,j} a^l_{ij}x^l_j - 1,
\end{equation*}
converge to an equilibrium, by interpreting them as dynamics of a primal-dual algorithm for solving~\eqref{eq:functional}.

The linear reward structure is readily generalized to concave homogenous rewards by following essentially the same proofs. In one-shot congestion games, inefficiency arises due to inhomogeneity of cost functions~\cite{roughgarden}.
In games exhibiting both endogenous delays and inhomogeneous rewards it will be interesting to determine how 
efficiency is affected by each.

The analysis in this paper may be useful in economic applications where both consumption and production of resources takes place. 
Activities with $a^l_{ij}<0$ can be thought\footnote{Note however that~\eqref{eq:decomposition1_mass} is no longer a mass constraint.}
 of as producing  $-a^l_{ij}$ units of resource $i$, while 
resources with $b_i<0$ are as if they are being discarded with rate $-b_i$.
Such models are considered in activity analysis, e.g., see~\cite{koopmans}, where the focus is in optimizing~\eqref{eq:primal}.
This can be a daunting task because the requirement of a centralized knowledge of production parameters is nonrealistic, and for this reason activity analysis has been subsumed by general equilibrium models~\cite{arrow}.
Nonetheless, an equilibrium concept in activity analysis, such as the one considered in
Definition~\ref{def:eq} and Theorem~\ref{thm:main}, may be useful in cases 
eluded by general equilibrium models.
Namely, cases where the production decisions are decentralized and taken by competing economic agents,
as in crowdsourced production
where prices may react slower to variations of supply and demand.


\bibliographystyle{plain}
\bibliography{biblio}

\end{document}